\documentclass[12pt]{article}
    \usepackage[margin=1in]{geometry}
    \usepackage[utf8]{inputenc}
    \usepackage[T1]{fontenc}    
    \usepackage{amsmath,amssymb,amsfonts,amsthm}
    \usepackage{mathtools}
    \usepackage[symbol]{footmisc}
    \usepackage{url}
    \usepackage{bigints}
    \usepackage{graphicx}
    \usepackage{algorithm2e}
    \usepackage{tabularx}
    \usepackage{xcolor}
    \usepackage{etoc}
    \usepackage{chngcntr}    
    \usepackage[figurename=Fig.]{caption}

\newtheorem{theorem}{Theorem}

\newtheorem{lemma}{Lemma}
\RestyleAlgo{ruled}

\begin{document}

\begin{titlepage}
\begin{center}

    \Large\textbf{Analytical characterisation of the Mi- and To-phases in HeMiTo dynamics: exponential growth and logistic saturation of toxic prion-like proteins}
       
    \normalsize
        
    \vspace{0.25cm}
    \setcounter{footnote}{0}
    \setlength{\footnotemargin}{0.8em}
    {\normalsize Johannes G Borgqvist\footnote{Corresponding author: \url{johborgq@chalmers.se}}\footnote{Mathematical Sciences, Chalmers University of Technology, Gothenburg, Sweden}}
    \setlength{\footnotemargin}{1.8em}
        
\abstract{Prion-like propagation of misfolded proteins is a key mechanism underlying the progression of neurodegenerative diseases such as Alzheimer’s disease. In previous work, we introduced the HeMiTo framework, describing these prion-like dynamics for a class of heterodimer models in terms of three phases: the healthy (He), mixed (Mi), and toxic (To) phases. While the He-phase was characterised analytically, the Mi-phase was described numerically and the To-phase was inferred from linear stability arguments.

In this work, we provide a complete analytical characterisation of the Mi- and To-phases for our class of heterodimer models. We derive exact inner solutions governing the Mi-phase and match them with outer solutions from the He-phase, explaining the concave-like behaviour of the healthy species and establishing explicit conditions for exponential growth of the toxic species with a mechanistically interpretable growth rate. Furthermore, we formalise a quasi steady-state reduction near the toxic steady state and show that the dynamics reduce to a logistic growth equation, linking exponential growth to saturation.

Together, these results provide a unified and mechanistic description of prion-like dynamics across all phases of disease progression and establish a foundation for predictive modelling of biomarker trajectories.}    
    
\end{center}
\end{titlepage}

\section{Introduction}
A critical component of many neurodegenerative diseases is the protein-equivalent of viruses known as \textit{prions}. In 1997, Stanley B. Prusiner in his Nobel lecture~\cite{prusiner1998prion} described them as follows.

\begin{quote}
    ``Prions are unprecedented infectious pathogens that cause a group of invariably fatal neurodegenerative diseases mediated by an entirely novel mechanism.''
\end{quote}
The then novel mechanism underlying all prion diseases involves a particular modification of the PrP prion protein, where the normal cellular form $\mathrm{PrP}^{\mathrm{C}}$ is converted by misfolding the protein into an infectious form $\mathrm{PrP}^{\mathrm{Sc}}$. The concept of protein misfolding is also central for the class of neurodegenerative diseases to which Alzheimer's disease belongs, which is called \textit{tauopathies}. Since tauopathies share many features with prion diseases and are centred on misfolding of tau-proteins, they are referred to as \textit{prion-like}~\cite{alyenbaawi2020prion}. To understand the underlying mechanisms of prion-like diseases, their infectious properties caused by protein misfolding have been the subject of mechanistic modelling.  

\begin{figure}[htbp!]
    \centering
    \includegraphics[width=0.9\linewidth]{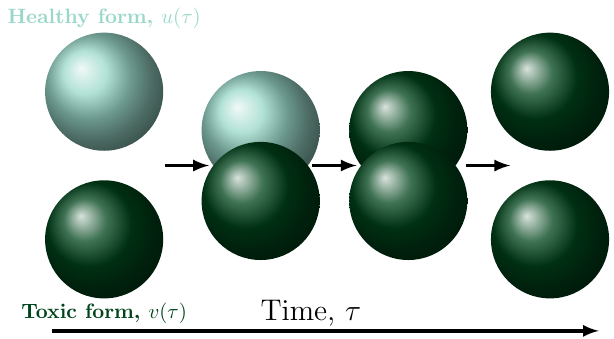}
    \caption{\textit{Infectious conversion of prion-like proteins}. When the toxic and infectious form $v(\tau)$ of the prion-like protein interacts with the healthy form $u(\tau)$, the latter gets converted into a toxic particle. This figure is adapted from Fig. 20 in~\cite{thompson2020protein}. }
    \label{fig:prion_reaction}
\end{figure}

The key reaction of prion-like diseases is the conversion of the healthy particle $u$ into its toxic and infectious counterpart $v$ (Fig. \ref{fig:prion_reaction}). Mechanistically speaking, the rate of this conversion reaction is described by the \textit{law of mass action} according to $r_{3}=\varepsilon{u}v$ where $\varepsilon$ is a rate parameter. In addition, a constant formation rate of the healthy species $r_{1}=c_{1}$ as well as linear degradation rates of each species $r_{2}=c_{2}u$ and $r_{4}=\varepsilon{v}$ are included in the simplest description of prion-like dynamics known as the \textit{heterodimer model}. Technically, the heterodimer model consists of a coupled system of two ODEs given by $\mathrm{d}u/\mathrm{d}\tau=r_{1}-r_{2}-r_{3}$ and $\mathrm{d}v/\mathrm{d}\tau=r_{3}-r_{4}$, and it has been used to model prion-like interactions of p-tau and amyloid $\beta$ in the brain~\cite{thompson2020protein} as well as the propagation of yeast prions~\cite{lemarre2020unifying}. Interestingly enough, the heterodimer model that describes the dynamics of infectious prions is also known as the \textit{SIR model with vital dynamics}~\cite{parsons:hal-04178969} in epidemiology, and thus this model structure is interesting in numerous applications. Recently, important generalisations of the heterodimer model have resulted in a class of models that describe the dynamics of distinct prion-like proteins. 

In our latest work, my co-author and I constructed a wider class of heterodimer models by generalising the term for the rate of the conversion reaction. As mentioned previously, the single toxic particle conversion rate (Fig. \ref{fig:prion_reaction}) is given by $r_{3}=\varepsilon{u}v$ and a two toxic particle conversion rate is given by $r_{3}=\varepsilon{u}v^{2}$. A general class of conversion rates is given by $r_{3}=uvf(v)$ defined by a general conversion function $f$, and mechanistically the choice of conversion function $f$ defines the mechanism underlying a specific prion-like protein of interest. Provided this framework, we constructed the following dimensionless class of heterodimer models~\cite{borgqvist2025HeMiTo}

\begin{align}    
\dfrac{\mathrm{d}u}{\mathrm{d}\tau}=&c_{1}-c_{2}u-\varepsilon{u}vf(v)\,,\label{eq:ODE_u}\\
\dfrac{\mathrm{d}v}{\mathrm{d}\tau}=&\varepsilon{v}\left(uf(v)-1\right)\,,\label{eq:ODE_v}\\
&u(\tau=0)=u_{0}\,,\quad{v(\tau=0)}=v_{0}\label{eq:IC_u_and_v}\,,
\end{align}
where $u_{0}$ and $v_{0}$ are the initial conditions of healthy and toxic species, respectively. Importantly, this dimensionless class of heterodimer models contains the perturbation parameter $\varepsilon\in(0,1)$ that allowed us to express solutions to this class of heterodimer models as series expansions of this perturbation parameter. Such perturbation series enable researchers to describe dynamics at different time scales, and here the initial dynamics corresponding to the $\mathcal{O}(1)$ terms is described by the \textit{outer solutions}, while the subsequent dynamics corresponding to the $\mathcal{O}(\varepsilon)$ terms is described by the \textit{inner solutions}~\cite{gerlee2022weak}. Using such a perturbation analysis, we described the dynamical features of prion-like proteins mathematically.

The prion-like dynamics of our class of heterodimer models are characterised by three phases captured by the acronym HeMiTo~\cite{borgqvist2025HeMiTo}. First, the system undergoes the healthy phase, known as the He-phase, where the concentration of the healthy species increases while the concentration of the toxic species remains constant. The system then enters the mixed phase, known as the Mi-phase, which marks the onset of disease progression as the toxic species increases at the expense of the healthy species. Lastly, the system undergoes the toxic phase, known as the To-phase, where the system evolves towards its toxic steady state.

In previous work, we derived exact analytical expressions for the outer solutions that describe the dynamics in the He-phase~\cite{borgqvist2025HeMiTo}. Moreover, we constructed numerical approximations of the inner solutions in the Mi-phase, suggesting that the healthy species follows a concave trajectory reaching a maximum, while the toxic species grows exponentially~\cite{borgqvist2025HeMiTo}. However, a rigorous analytical characterisation of the Mi-phase, as well as a corresponding analysis of the To-phase dynamics, remains an open problem.

In this work, we address these questions in four parts. First, we derive exact analytical expressions for the inner solutions governing the Mi-phase and match them with the outer solutions from the preceding He-phase. Second, we use these matched asymptotics to explain the perceived concavity of the healthy species and identify the conditions under which it attains a maximum. Third, we establish precise conditions under which the toxic species is well approximated by an exponential function with a mechanistically interpretable growth rate. Finally, we analyse the To-phase by showing that, under a quasi steady-state approximation, the dynamics reduce to a logistic growth equation that captures saturation towards the toxic steady state. Before presenting these results in Section~\ref{sec:res}, we briefly introduce the mathematical preliminaries in Section~\ref{sec:math}.

\section{Mathematical preliminaries}\label{sec:math}
Here, we describe the results of our previous article~\cite{borgqvist2025HeMiTo} on outer solutions describing the He-phase of our class of heterodimer models in Eqs. \eqref{eq:ODE_u} to \eqref{eq:IC_u_and_v}. These results provide a natural starting point for the inner solutions describing the Mi-phase in this work presented subsequently in Section \ref{sec:res}.

Biologically feasible dynamics\footnote{On this note on biological feasibility, all parameters are assumed to be positive and the states are assumed to be non-negative.} of the class of heterodimer model are ensured by the existence of two distinct steady states. First, the system of ODEs in Eqs. \eqref{eq:ODE_u} and \eqref{eq:ODE_v} has a \textit{healthy steady state} (HSS) given by
\begin{equation}
    \mathrm{HSS}=\left(u_{1}^{\star},v_{1}^{\star}\right)=\left(\frac{c_{1}}{c_{2}},0\right)\,.
    \label{eq:HSS}
\end{equation}
Second, if there exists a $v_{2}^{\star}>0$ such that $f(v_{2}^{\star})>0$ that solves
\begin{equation}
    v_{2}^{\star}=\frac{1}{\varepsilon}\left(c_{1}-\frac{c_{2}}{f(v_{2}^{\star})}\right)\,,
    \label{eq:v_2_star}
\end{equation}
then the system has a \textit{toxic steady state} (TSS) given by
\begin{equation}
    \mathrm{TSS}=\left(u_{2}^{\star},v_{2}^{\star}\right)=\left(\frac{1}{f(v_{2}^{\star})},\frac{1}{\varepsilon}\left(c_{1}-\frac{c_{2}}{f(v_{2}^{\star})}\right)\right)\,.
    \label{eq:TSS}
\end{equation}
Provided these steady states, biologically feasible dynamics correspond to the \textit{HSS being a saddle point} and the \textit{TSS being a stable node}. These dynamical properties are neatly expressed in terms of the conversion function $f$ together with the function $g$ given by
\begin{equation}
g(v)=\frac{c_{2}}{c_{1}-\varepsilon{v}}\,,
    \label{eq:g}
\end{equation}
and mathematically these properties correspond to the following conditions:
    \begin{equation}
        f(v)>{g(v)}\,\forall{v}\in[0,v_{2}^{\star})\,,\quad{f}(v_{2}^{\star})=g(v_{2}^{\star})\,,\quad{f}'(v_{2}^{\star})\leq{0}\,.
        \label{eq:f_bound_by_g}
    \end{equation}
For the perturbation analysis, we impose assumptions on the initial conditions in Eq. \eqref{eq:IC_u_and_v}: $u_{0},v_{0}=\mathcal{O}(1)$ and the conversion function $f$: $f(v_{0})=\mathcal{O}(1)$ and $f$ is analytical. Given these assumptions, we construct the following perturbation ans\"{a}tze 
\begin{align}
    u(\tau)&=u_{\mathrm{He}}(\tau)+u_{\mathrm{Mi}}(\tau)\varepsilon+\mathcal{O}(\varepsilon^{2})\,,\label{eq:u_approx}\\
    v(\tau)&=v_{\mathrm{He}}(\tau)+v_{\mathrm{Mi}}(\tau)\varepsilon+\mathcal{O}(\varepsilon^{2})\,,\label{eq:v_approx}
\end{align}
where the outer solutions $u_{\mathrm{He}}$ and $v_{\mathrm{He}}$ describe the dynamics in the He-phase and the inner solutions $u_{\mathrm{Mi}}$ and $v_{\mathrm{Mi}}$ describe the dynamics in the Mi-phase. By the analyticity of the conversion function $f$, we have
\begin{equation}
    f\left(v\right)=f\left(v_{\mathrm{He}}+v_{\mathrm{Mi}}\varepsilon+\mathcal{O}(\varepsilon^{2})\right)=f\left(v_{\mathrm{He}}\right)+\mathcal{O}(\varepsilon)\,.
    \label{eq:f_ana}
\end{equation}
In our previous work, we found the outer solutions dictating the dynamics in the initial He-phase, and these are given by~\cite{borgqvist2025HeMiTo}:
\begin{align}
    u_{\mathrm{He}}(\tau)&=\frac{c_{1}}{c_{2}}-\left(\frac{c_{1}}{c_{2}}-u_{0}\right)e^{-c_{2}\tau}\,,\label{eq:u_He}\\
    v_{\mathrm{He}}(\tau)&=v_{0}\,.\label{eq:v_He}
\end{align}
Next, we find the corresponding inner solutions in order to analytically characterise the dynamics in the subsequent Mi-phase.

\section{Results}\label{sec:res}

Collectively, our results provide a complete analytical characterisation of the Mi- and To-phases of the HeMiTo dynamics framework. We show that the healthy species exhibits a concave-like trajectory, attaining a well-defined maximum, while the toxic species transitions from its initial, constant concentration in the He-phase into exponential growth during the Mi-phase and subsequently into saturation in the To-phase.

To establish these results, we first derive the inner solutions that govern the Mi-phase dynamics. These are then matched with the outer solutions from the He-phase to obtain an analytical expression for the maximum concentration of the healthy species. We further identify precise conditions under which the toxic species exhibits exponential growth in the Mi-phase, yielding a mechanistically interpretable growth rate. Finally, we show that, when the system is initialised near the TSS in the To-phase, the dynamics reduce to a logistic growth equation, thereby linking exponential growth in the Mi-phase to saturation in the To-phase.

\subsection{Mi-phase dynamics are governed by inner solutions}
The inner solutions governing the dynamics of the Mi-phase are obtained by solving simpler ODEs that originate from the $\mathcal{O}(\varepsilon)$ terms in the perturbation ans\"{a}te. Specifically, the perturbation ans\"{a}tze in Eqs. \eqref{eq:u_approx} to \eqref{eq:f_ana} are substituted into the class of heterodimer models in Eqs. \eqref{eq:ODE_u} to \eqref{eq:IC_u_and_v} and then the $\mathcal{O}(\varepsilon)$ terms are extracted from the resulting equations. To match our inner solutions with the outer solutions obtained previously in Eqs. \eqref{eq:u_He} and \eqref{eq:v_He}, we impose the following initial conditions on the inner solutions:
\begin{align}
    u_{\mathrm{Mi}}(\tau=0)&=\lim_{\tau\to\infty}u_{\mathrm{He}}(\tau)=\dfrac{c_{1}}{c_{2}}\,,\label{eq:u_Mi_IC}\\
    v_{\mathrm{Mi}}(\tau=0)&=\lim_{\tau\to\infty}v_{\mathrm{He}}(\tau)=v_{0}\,.\label{eq:v_Mi_IC}
\end{align}
Using these initial conditions, we find analytical equations for the inner solutions characterising the Mi-phase dynamics (Theorem \ref{thm:inner_sol}).

\begin{theorem}[Inner solutions $u_{\mathrm{Mi}}(\tau)$ and $v_{\mathrm{Mi}}(\tau)$]
Consider the perturbation ans\"{a}tze for the healthy species in Eq. \eqref{eq:u_approx} and the toxic species in Eq. \eqref{eq:v_approx} for the solutions of the class of heterodimer models in Eqs. \eqref{eq:ODE_u} to \eqref{eq:IC_u_and_v}. Furthermore, assume that the initial conditions in Eq. \eqref{eq:IC_u_and_v} satisfy $u_{0},v_{0}=\mathcal{O}(1)$ and that the conversion function $f$ is analytical and satisfies $f(v_{0})=\mathcal{O}(1)$. Then the inner solutions are given by
\begin{align}
u_{\mathrm{Mi}}(\tau)
&=
\left(\frac{c_1}{c_2}+\frac{c_1v_0f(v_0)}{c_2^2}\right)e^{-c_2\tau}
+v_0f(v_0)\left(\frac{c_1}{c_2}-u_0\right)\tau e^{-c_2\tau}
-\frac{c_1v_0f(v_0)}{c_2^2},\label{eq:u_MI}
\\[0.5em]
v_{\mathrm{Mi}}(\tau)
&=
v_0-\frac{v_0f(v_0)}{c_2}\left(\frac{c_1}{c_2}-u_0\right)
+v_0\left(f(v_0)\frac{c_1}{c_2}-1\right)\tau
+\frac{v_0f(v_0)}{c_2}\left(\frac{c_1}{c_2}-u_0\right)e^{-c_2\tau}.\label{eq:v_Mi}
\end{align}
which satisfy the initial conditions in Eqs. \eqref{eq:u_Mi_IC} and \eqref{eq:v_Mi_IC}, respectively.
    \label{thm:inner_sol}
\end{theorem}
\begin{proof}
    We substitute the outer solution $u_{\mathrm{He}}$ in Eq. \eqref{eq:u_He} into the perturbation ansatz for $u$ in Eq. \eqref{eq:u_approx} and the outer solution $v_{\mathrm{He}}$ in Eq. \eqref{eq:v_He} into the perturbation ansatz for $v$ in Eq. \eqref{eq:v_approx}. Thereafter, the resulting two ans\"{a}tze together with the expansion for the conversion function $f$ in Eq. \eqref{eq:f_ana} are substituted into the original ODEs in Eqs. \eqref{eq:ODE_u} and \eqref{eq:ODE_v}. In the resulting equations, we extract the $\mathcal{O}(\varepsilon)$ terms which yields the following decoupled system of ODEs for the inner solutions:
    \begin{align}
    \dfrac{\mathrm{d}u_{\mathrm{Mi}}}{\mathrm{d}\tau}&=-c_{2}u_{\mathrm{Mi}}-v_{0}f(v_{0})\left(\dfrac{c_{1}}{c_{2}}-\left(\dfrac{c_{1}}{c_{2}}-u_{0}\right)\exp(-c_{2}\tau)\right)\,,\label{eq:u_Mi_ODE}\\
    \dfrac{\mathrm{d}v_{\mathrm{Mi}}}{\mathrm{d}\tau}&=v_{0}\left(f(v_{0})\left(\dfrac{c_{1}}{c_{2}}-\left(\dfrac{c_{1}}{c_{2}}-u_{0}\right)\exp(-c_{2}\tau)\right)-1\right)\,.\label{eq:v_Mi_ODE}
\end{align}
In reverse order, the ODE in Eq. \eqref{eq:v_Mi_ODE} is separable and is solved directly by integration and the ODE in Eq. \eqref{eq:u_Mi_ODE} is solved by means of an integrating factor. The inner solutions are therefore given by:
\begin{align}
u_{\mathrm{Mi}}(\tau)
&=
K_{1}e^{-c_{2}\tau}
+v_{0}f(v_{0})\left(\frac{c_{1}}{c_{2}}-u_{0}\right)\tau e^{-c_{2}\tau}
-\frac{c_{1}v_{0}f(v_{0})}{c_{2}^{2}}\,,
\label{eq:u_Mi_first}\\[0.5em]
v_{\mathrm{Mi}}(\tau)
&=
K_{2}
+v_{0}\left(f(v_{0})\frac{c_{1}}{c_{2}}-1\right)\tau
+\frac{v_{0}f(v_{0})}{c_{2}}\left(\frac{c_{1}}{c_{2}}-u_{0}\right)e^{-c_{2}\tau}\,,\label{eq:v_Mi_first}
\end{align}
where $K_{1}$ and $K_{2}$ are two arbitrary integration constants. The value of these integration constants that satisfy the initial conditions in Eqs. \eqref{eq:u_Mi_IC} and \eqref{eq:v_Mi_IC} are given by
\begin{align}
K_1 &= \frac{c_1}{c_2}+\frac{c_1v_0f(v_0)}{c_2^2}\,,
\\
K_2 &= v_0-\frac{v_0f(v_0)}{c_2}\left(\frac{c_1}{c_2}-u_0\right)\,,
\end{align}
which completes the proof.
\end{proof}
Using these inner solutions, we begin by investigating the apparent concavity, discovered in our previous work~\cite{borgqvist2025HeMiTo}, of the healthy species during the He- and Mi-phases.

\subsection{Asymptotic matching in the Mi-phase yields the healthy steady state maximum}
The asymptotic matching of the inner and outer solutions for the healthy species (Fig. \ref{fig:matching}) shows that the maximum of the healthy species is, in fact, given by its HSS value $c_{1}/c_{2}$ according to Eq. \eqref{eq:HSS}. Provided a biologically realistic initial condition for the healthy species $u_{0}\in(0,c_{1}/c_{2})$, the derivative of the outer solution is positive:
\begin{equation}
    \dfrac{\mathrm{d}u_{\mathrm{He}}}{\mathrm{d}\tau}=\left(\frac{c_{1}}{c_{2}}-u_{0}\right)e^{-c_{2}\tau}>0\,.
\label{eq:u_He_increasing}
\end{equation}
Hence, $u_{\mathrm{He}}$ is an increasing function where its maximum value is given by its HSS-value according to $\max_{\tau\in[0,\infty)}u_{\mathrm{He}}(\tau)=\lim_{\tau\to\infty}u_{\mathrm{He}}(\tau)=c_{1}/c_{2}$. 
\begin{figure}[ht!]
    \centering
    \includegraphics[width=\linewidth]{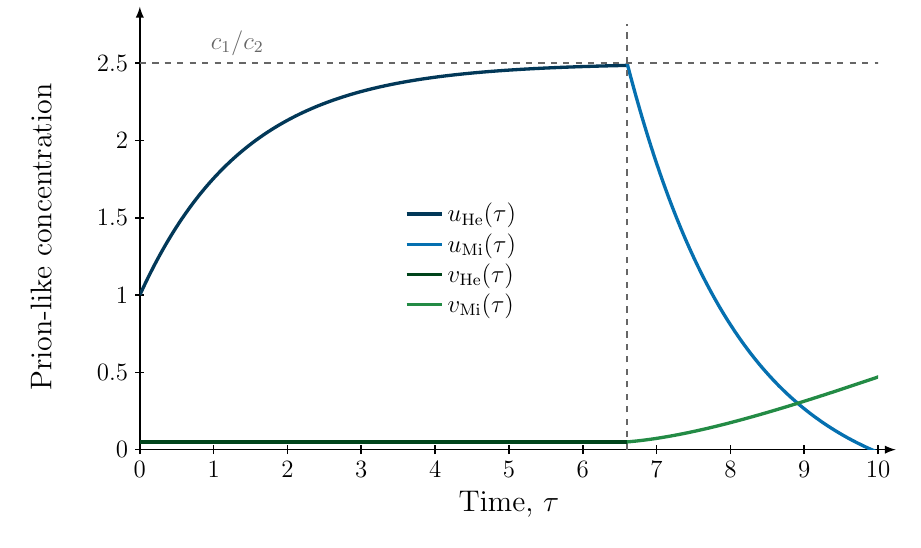}
    \caption{\textit{Asymptotic matching of the inner and outer solutions}. The inner and outer solutions for the healthy species denoted by $u_{\mathrm{He}}$ and $u_{\mathrm{Mi}}$ are illustrated in blue curves and bounded by the maximum value $c_{1}/c_{2}$, while the inner and outer solutions for the toxic species denoted by $v_{\mathrm{He}}$ and $v_{\mathrm{Mi}}$ are illustrated in the green curves. The parameters and initial conditions for the illustrated curves are $c_{1}=1.75$, $c_{2}=0.70$, $u_{0}=1.00$, $v_{0}=0.05$ and $f(v_{0})=1.80$. }
    \label{fig:matching}
\end{figure}

Using a similar line of reasoning, the derivative of the inner solution for the heathy species $u_{\mathrm{Mi}}$ is given by
\begin{align}
\dfrac{\mathrm{d}u_{\mathrm{Mi}}}{\mathrm{d}\tau}&=-e^{-c_2\tau}\left((c_1+v_0f(v_0)u_0)+c_2v_0f(v_0)\left(\frac{c_1}{c_2}-u_0\right)\tau\right)\,,
\end{align}
which means that it is decreasing, i.e. $\mathrm{d}u_{\mathrm{Mi}}/\mathrm{d}\tau<0$. Its maximum value is therefore given by its initial condition $\max_{\tau\in[0,\infty)}u_{\mathrm{Mi}}(\tau)=u_{\mathrm{Mi}}(\tau=0)=c_{1}/c_{2}$, which again corresponds to the HSS-value of the healthy species. Thus, the maximum value of the healthy species is given by its HSS-value which corresponds to the point where the inner and outer solutions are matched (Fig. \ref{fig:matching}). 

In total, the matching of the inner and outer solutions for the healthy species explains the apparent concavity of the trajectory for the healthy species observed in our previous work~\cite{borgqvist2025HeMiTo} obtained by numerically solving the original class of heterodimer models. Dynamically speaking, the outer solution approaches its HSS-value during the He-phase due to the fact that the constant formation rate $c_{1}$ drives the dynamics, and then when the Mi-phase starts the degradation and conversion rates drive the dynamics instead, which causes a decrease in the concentration of the healthy species. Consequently, the function value for which the outer and inner solutions for the healthy species are matched, which is the HSS value $c_{1}/c_{2}$, corresponds to the maximum value. Another numerically motivated feature observed in our previous work~\cite{borgqvist2025HeMiTo} is the exponential growth of the toxic species, which we analyse analytically using our matched asymptotics next.

\subsection{Mi-phase dynamics yield exponential growth of the toxic species}
Our main result is that the toxic species grows exponentially during the Mi-phase, corresponding to the onset of disease, \textit{under a specific condition}. This condition is that the initial concentration of the healthy species, $u_{0}$, must be ``sufficiently high'', and we can quantify using our analytical inner solution for the toxic species what ``sufficiently high'' means exactly. Thus, given a sufficiently high initial concentration $u_{0}$, the toxic species $v_{\mathrm{Mi}}(\tau)$ grows exponentially during the onset of disease (Theorem \ref{thm:exp}).

\begin{theorem}[Exponential approximation of the toxic species]
Let
\begin{equation}
v_{\mathrm{Mi}}(\tau)
=
v_0
+v_0\left(f(v_0)\frac{c_1}{c_2}-1\right)\tau
+\frac{v_0f(v_0)}{c_2}
\left(\frac{c_1}{c_2}-u_0\right)\left(e^{-c_2\tau}-1\right),
\label{eq:v_Mi_thm_2}
\end{equation}
where $v_0$, $f(v_0)$, $c_1$, $c_2$, and $u_0$ are positive constants. Define
\begin{equation}
\lambda:=f(v_0)u_0-1.
\label{eq:lambda}
\end{equation}
Then, as $\tau\to 0$,
\begin{equation}
v_{\mathrm{Mi}}(\tau)=v_0 e^{\lambda\tau}+\mathcal{O}(\tau^2)\,.
\end{equation}
Hence, for sufficiently small times, $v_{\mathrm{Mi}}$ is well-approximated by the exponential function
\begin{equation}
v_0 e^{(f(v_0)u_0-1)\tau}.
\label{eq:exp_approx}
\end{equation}
\label{thm:exp}
\end{theorem}
\begin{proof}
Using the Taylor expansion of the exponential function,
\begin{equation}
e^{-c_2\tau}
=1-c_2\tau+\frac{c_2^2\tau^2}{2}-\frac{c_2^3\tau^3}{6}+\mathcal{O}(\tau^4),
\qquad \tau\to 0,
\end{equation}
in Eq. \eqref{eq:v_Mi_thm_2} for $v_{\mathrm{Mi}}$, we obtain
\begin{align}
v_{\mathrm{Mi}}(\tau)
&=
v_0+v_0\left(f(v_0)\frac{c_1}{c_2}-1\right)\tau
+\frac{v_0f(v_0)}{c_2}
\left(\frac{c_1}{c_2}-u_0\right)
\left(
-c_2\tau+\frac{c_2^2\tau^2}{2}
\right)+\mathcal{O}(\tau^3)
\nonumber\\
&=v_0
+v_0\bigl(\underset{=\lambda}{\underbrace{f(v_0)u_0-1}}\bigr)\tau
+\frac{v_0f(v_0)(c_1-c_2u_0)}{2}\tau^2
+\mathcal{O}(\tau^3)\nonumber\,,
\end{align}
which yields that 
\begin{equation}
    v_{\mathrm{Mi}}(\tau)=v_0+v_0\lambda\tau
+\frac{v_0f(v_0)(c_1-c_2u_0)}{2}\tau^2+O(\tau^3)\,.
\label{eq:expansion_1}
\end{equation}
On the other hand,
\begin{equation}
v_0e^{\lambda\tau}=v_0
+v_0\lambda\tau
+\frac{v_0\lambda^2}{2}\tau^2
+\mathcal{O}(\tau^3),\qquad \tau\to 0.
\label{eq:expansion_2}
\end{equation}
Subtracting the expansions in Eqs. \eqref{eq:expansion_1} and \eqref{eq:expansion_2} yields
\begin{equation}
v_{\mathrm{Mi}}(\tau)-v_0e^{\lambda\tau}
=\frac{v_0}{2}
\left(f(v_0)(c_1-c_2u_0)-\lambda^2
\right)\tau^2
+\mathcal{O}(\tau^3)=\mathcal{O}(\tau^{2}),
\end{equation}
which proves the claim.
\end{proof}

\begin{figure}[ht!]
\centering
\includegraphics[width=\linewidth]{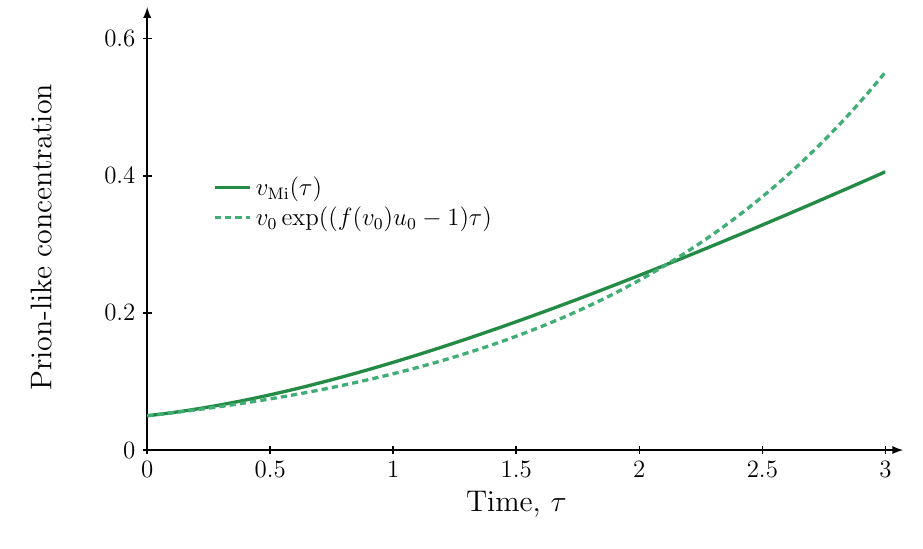}
\caption{\textit{Exponential approximation of the toxic species during the Mi-phase}. The inner solution $v_{\mathrm{Mi}}(\tau)$ in Eq. \eqref{eq:v_Mi_thm_2}, the solid line, is compared to the exponential approximation $v_{0}\exp\left((u_{0}f(v_{0})-1)\tau\right)$ in Eq. \eqref{eq:exp_approx}, the dashed line. The exponential approximation is accurate for early time points when $\tau\in[0,0.5]$. The parameters and initial conditions for the illustrated curves are $c_{1}=1.75$, $c_{2}=0.70$, $u_{0}=1.00$, $v_{0}=0.05$ and $f(v_{0})=1.80$. }
    \label{fig:exp_approx}
\end{figure}

To guarantee exponential growth of the toxic species during the Mi-phase, we require that the growth rate $\lambda$ in Eq. \eqref{eq:lambda} is positive. The positivity condition $\lambda>0$ is equivalent to the following inequality
\begin{equation}
    u_{0}>\dfrac{1}{f(v_{0})}\,,
    \label{eq:u0_large}
\end{equation}
and thus the toxic species grows exponentially during the Mi-phase when the initial concentration $u_{0}$ satisfies this lower bound. When this lower bound is satisfied, the growth rate of the toxic species satisfies $\lambda\propto u_{0}$ and $\lambda\propto f(v_{0})$. These two proportionalities imply that the growth rate of the toxic species during the onset of disease is directly determined by both the initial concentration of the healthy species and the conversion function $f$ encoding the mechanistic features of the prion-like protein of interest.

The exponential growth derived above characterises the Mi-phase, where the toxic species increases rapidly from low concentrations. We now turn to the subsequent To-phase, where the system evolves in a neighbourhood of the toxic steady state (TSS) in Eq. \eqref{eq:TSS}. In this regime, the dynamics are no longer governed by exponential growth but instead by saturation. To capture this saturation, we revisit physically motivated arguments and formalise them within our framework, showing that the dynamics reduce to a logistic growth law.

\subsection{To-phase dynamics through quasi steady-state reduction yield logistic growth}
We now turn to the final phase of the dynamics, namely the toxic (To) phase, where the system evolves in a neighbourhood of the toxic steady state (TSS). Previous work has shown that, in this regime, the dynamics should be well approximated by the linearised system around the TSS~\cite{borgqvist2025HeMiTo,gerlee2022weak}. In parallel, Fornari et al.~\cite{fornari2019prion} proposed a physically motivated argument, suggesting that when the healthy species dominates, i.e.~$u \gg v$, the dynamics reduce to a logistic growth equation for the toxic species. Our aim here is to formalise this physically motivated argument within our class of heterodimer models and to derive the logistic growth law rigorously from the governing equations.

We begin by clarifying what is meant by a quasi steady-state of the healthy species near the TSS.

\begin{lemma}[Quasi steady-state near the toxic steady state]
Consider the heterodimer model in Eqs.~(2.8) and (2.9). Suppose that the toxic steady state 
$\mathrm{TSS} = (u_{2}^{\star},v_{2}^{\star})$ in Eq.~\eqref{eq:TSS} exists, that it is stable in the sense of Eq. \eqref{eq:f_bound_by_g} and assume that the solution is initialised at
\begin{align}
u_0 = u_2^\star.
\end{align}
Then, for all $\tau \geq 0$, the evolution of the healthy species satisfies
\begin{equation}
\frac{\mathrm{d}u}{\mathrm{d}\tau}
=\varepsilon\left(
v_2^\star - \frac{v f(v)}{f(v_2^\star)}
\right).
\label{eq:u_TSS}
\end{equation}
In particular, if $v$ lies in a neighbourhood of $v_2^\star$ such that
\begin{align}
v_2^\star - \frac{v f(v)}{f(v_2^\star)} = \mathcal{O}(1),
\end{align}
then
\begin{align}
\frac{\mathrm{d}u}{\mathrm{d}\tau} = \mathcal{O}(\varepsilon),
\label{eq:u_time_scale}
\end{align}
and the healthy species is quasi-stationary at early times.
\label{lemma:quasi}
\end{lemma}

\begin{proof}
From the governing equation for the healthy species, we have
\begin{align}
\frac{\mathrm{d}u}{\mathrm{d}\tau}
=
c_1 - c_2 u - \varepsilon u v f(v).
\end{align}
Using the assumption $u_0 = u_2^\star$ together with the expression for the toxic steady state,
\begin{align}
u_2^\star = \frac{1}{f(v_2^\star)},
\end{align}
we obtain
\begin{align}
\frac{\mathrm{d}u}{\mathrm{d}\tau}
=
c_1 - \frac{c_2}{f(v_2^\star)} - \varepsilon \frac{v f(v)}{f(v_2^\star)}.
\end{align}
By definition of the toxic steady state $v_{2}^{\star}$, we have
\begin{align}
c_1 - \frac{c_2}{f(v_2^\star)} = \varepsilon v_2^\star,
\end{align}
which yields
\begin{align}
\frac{\mathrm{d}u}{\mathrm{d}\tau}
=
\varepsilon\left(
v_2^\star - \frac{v f(v)}{f(v_2^\star)}
\right).
\end{align}
The stated estimate follows immediately.
\end{proof}
In particular, when $v = v_2^\star$ in Eq. \eqref{eq:u_TSS}, we have $\mathrm{d}u/\mathrm{d}\tau = 0$, so that the healthy species is exactly stationary at the TSS. Thus, the quasi steady-state approximation is exact at the TSS and remains accurate in its neighbourhood.

Lemma~\ref{lemma:quasi} shows that, when the system is initialised in a neighbourhood of the TSS, the healthy species evolves on a slow time scale. More precisely, the rate of change of the healthy species is proportional to the perturbation parameter $\varepsilon$ according to Eq.~\eqref{eq:u_time_scale}. Since $\varepsilon < 1$, the evolution of $u$ is slow compared to that of $v$. The healthy species is therefore naturally approximated by a quasi steady-state.

Under this quasi steady-state approximation, we set $\mathrm{d}u/\mathrm{d}\tau \approx 0$ in the first equation, which yields the algebraic relation
\begin{align}
u = \frac{c_1}{c_2 + \varepsilon v f(v)}\,,
\end{align}
for the healthy species. Substituting this expression into the equation for the toxic species gives the scalar equation
\begin{equation}
\frac{\mathrm{d}v}{\mathrm{d}\tau}
=\varepsilon v \left( \frac{c_1 f(v)}{c_2 + \varepsilon v f(v)} - 1 \right).
\label{eq:scalar_v}
\end{equation}
We are now in a position to derive the logistic growth law (Theorem \ref{thm:logistic}).

\begin{theorem}[Local logistic growth of the toxic species near the TSS]
Consider the class of heterodimer models in Eqs.~(1) and (2) subject to positive initial
conditions in Eq.~(3), and assume that the conversion function $f$ is analytic. Suppose that
the toxic steady state $\mathrm{TSS} = (u_{2}^{\star},v_{2}^{\star})$ in Eq. \eqref{eq:TSS} exists and that it is stable in the sense of Eq. \eqref{eq:f_bound_by_g} . Assume furthermore that the system is initialised in a quasi steady-state in the sense of Lemma~\ref{lemma:quasi}, i.e.~$u_0 = u_2^\star$, with $v_0$ lying in a neighbourhood of $v_2^\star$. Then, as $\tau\to 0$, the toxic species
is approximated by a logistic growth function
\begin{equation}
v(\tau)=
\frac{\kappa}{1+\left(\frac{\kappa-v_{0}}{v_{0}}\right)e^{-\gamma\tau}}
+\mathcal{O}(\tau^{3}),
\label{eq:logistic}
\end{equation}
where the growth rate $\gamma$ and the carrying capacity $\kappa$ are given by
\begin{align}
\gamma
&=\varepsilon\left(\Psi(v_{0})-v_{0}\Psi'(v_{0})\right),\label{eq:growth_rate}
\\[0.5em]
\kappa
&=
v_{0}-\frac{\Psi(v_{0})}{\Psi'(v_{0})}\label{eq:carrying_capacity},
\end{align}
provided $\Psi'(v_{0})\neq 0$, and where
\begin{align}
\Psi(v):=\left(\frac{c_{1}f(v)}{c_{2}+\varepsilon v f(v)}-1\right).
\label{eq:psi}
\end{align}
Consequently, the toxic species is locally approximated by a logistic growth function.
\label{thm:logistic}
\end{theorem}
\begin{proof}
Under the quasi steady-state approximation justified in Lemma~\ref{lemma:quasi}, the healthy
species is given by
\begin{align}
u(v)=\frac{c_{1}}{c_{2}+\varepsilon v f(v)}.
\end{align}
Substituting this expression into the ODE for the toxic species yields the scalar equation
\begin{align}
\frac{\mathrm{d}v}{\mathrm{d}\tau}&=\varepsilon\left(uvf(v)-v\right)\nonumber\\
&=\varepsilon v\left(\frac{c_{1}f(v)}{c_{2}+\varepsilon v f(v)}-1\right)\nonumber\\
&=\varepsilon v\,\Psi(v).\label{eq:scalar_v_proof}
\end{align}
Since the conversion function $f$ is analytic, the function $\Psi$ in Eq.~\eqref{eq:psi}
is also analytic in a neighbourhood of $v_{0}$, and therefore
\begin{align}
\Psi(v)
=
\Psi(v_{0})+\Psi'(v_{0})(v-v_{0})+\mathcal{O}\!\left((v-v_{0})^{2}\right).
\end{align}
Substituting this expansion into the scalar equation in Eq. \eqref{eq:scalar_v_proof} gives
\begin{align}
\frac{\mathrm{d}v}{\mathrm{d}\tau}
&=
\varepsilon v\left[\Psi(v_{0})+\Psi'(v_{0})(v-v_{0})+\mathcal{O}\!\left((v-v_{0})^{2}\right)\right]
\nonumber\\
&=
\varepsilon v\left[\Psi(v_{0})-v_{0}\Psi'(v_{0})+\Psi'(v_{0})v+\mathcal{O}\!\left((v-v_{0})^{2}\right)\right].
\end{align}
Rearranging the leading-order terms yields
\begin{equation}
\frac{\mathrm{d}v}{\mathrm{d}\tau}=\gamma v\left(1-\frac{v}{\kappa}\right)+\varepsilon v\,\mathcal{O}\!\left((v-v_{0})^{2}\right),
\label{eq:scalar_v_proof_2}
\end{equation}
where
\begin{align}
\gamma
&=
\varepsilon\left(\Psi(v_{0})-v_{0}\Psi'(v_{0})\right),
\\[0.5em]
\kappa
&=
v_{0}-\frac{\Psi(v_{0})}{\Psi'(v_{0})},
\end{align}
provided $\Psi'(v_{0})\neq 0$. Next, we express the error term in Eq. \eqref{eq:scalar_v_proof_2} in terms of time $\tau$ instead of the difference $v-v_{0}$. 

Since $v$ is differentiable and $v(0)=v_{0}$, we have
\begin{align}
v(\tau)-v_{0}=\left.\frac{\mathrm{d}v}{\mathrm{d}\tau}\right|_{\tau=0}\tau+\mathcal{O}(\tau^{2}),
\qquad \tau\to 0,
\end{align}
and therefore
\begin{align}
(v(\tau)-v_{0})^{2}=\mathcal{O}(\tau^{2}).
\end{align}
Since $v$ lies in a neighbourhood of $v_2^\star$ and $v_2^\star=\mathcal{O}(1/\varepsilon)$,
we have
\begin{align}
\varepsilon v=\mathcal{O}(1).
\end{align}
It follows that
\begin{align}
\varepsilon v\,\mathcal{O}\!\left((v-v_{0})^{2}\right)=\mathcal{O}\!\left((v-v_{0})^{2}\right)=\mathcal{O}(\tau^{2}),
\end{align}
so the scalar equation in Eq. \eqref{eq:scalar_v_proof_2} reduces to
\begin{align}
\frac{\mathrm{d}v}{\mathrm{d}\tau}
=
\gamma v\left(1-\frac{v}{\kappa}\right)+\mathcal{O}(\tau^{2}).
\end{align}
Neglecting the higher-order term, we obtain the logistic growth equation
\begin{align}
\frac{\mathrm{d}v}{\mathrm{d}\tau}
=
\gamma v\left(1-\frac{v}{\kappa}\right),
\end{align}
and integrating with respect to $\tau$ yields
\begin{align}
v(\tau)
=
\frac{\kappa}{1+\left(\frac{\kappa-v_{0}}{v_{0}}\right)e^{-\gamma\tau}}.
\end{align}
Since the neglected term is of order $\mathcal{O}(\tau^{2})$, the resulting solution satisfies
\begin{align}
v(\tau)
=
\frac{\kappa}{1+\left(\frac{\kappa-v_{0}}{v_{0}}\right)e^{-\gamma\tau}}
+
\mathcal{O}(\tau^{3}),
\qquad \tau\to 0,
\end{align}
which proves the claim.
\end{proof}

This result places the logistic growth law in context. Far from the carrying capacity, i.e.~when $v \ll \kappa$, the quadratic term becomes negligible and the dynamics reduce to exponential growth. In contrast, as $v$ approaches $\kappa$, the nonlinear term dominates and the growth saturates. Thus, exponential and logistic growth appear as two complementary regimes of the same underlying dynamics, depending on the proximity of the system to the TSS.

The logistic approximation in Theorem~\ref{thm:logistic} also admits a clear mechanistic interpretation. Near the TSS, the carrying capacity satisfies $\kappa \approx v_2^\star$, so that the logistic law saturates towards the TSS concentration of the original heterodimer model. Thus, saturation in the logistic approximation corresponds to convergence towards the TSS of the full system. In contrast to the Mi-phase, where the toxic species grows exponentially with growth rate $\lambda = f(v_0)u_0 - 1$, the To-phase is therefore characterised by saturation rather than rapid growth. Indeed, logistic growth reduces to exponential growth only when $v \ll \kappa$, whereas in the To-phase we assume that $v$ lies in a neighbourhood of $v_2^\star \approx \kappa$.

Taken together, these results show that the Mi- and To-phases describe two distinct dynamical regimes: exponential take-off away from the HSS and logistic saturation towards the TSS. This distinction provides a unified interpretation of the HeMiTo dynamics framework and forms the basis for the discussion of their biological implications in the following section.

\section{Discussion}\label{sec:dis}
Together with our previous work~\cite{borgqvist2025HeMiTo}, the present study provides a complete analytical characterisation of prion-like toxicity within the HeMiTo dynamics framework. In our earlier work, we introduced a non-dimensionalisation yielding a natural perturbation parameter, used linear stability analysis to identify conditions under which the toxic species exhibits sigmoidal growth, and proposed the HeMiTo framework to describe the dynamics qualitatively. While the He-phase was characterised analytically, the Mi-phase was only described numerically and the To-phase was inferred from linear stability arguments.

In the present work, we close this gap. We derive exact inner solutions for the Mi-phase and match them to the outer solutions from the He-phase, thereby explaining the apparent concave-like trajectory of the healthy species and identifying its maximum given by the HSS value $c_{1}/c_{2}$. Moreover, we establish explicit and mechanistically interpretable conditions for exponential growth of the toxic species in the Mi-phase, namely $u_0 > 1/f(v_0)$, with growth rate $f(v_0)u_0 - 1$. Finally, we formalise a previously heuristic argument by Fornari et al.~\cite{fornari2019prion} and show that, in the To-phase near the TSS, the dynamics reduce to a logistic growth law with carrying capacity given by the TSS value for the toxic species. Taken together, these results reveal that the HeMiTo dynamics consist of two qualitatively distinct regimes: the initial He- and Mi-phases centred around the HSS characterised by exponential take-off of the toxic species driven by prion-like conversion of the healthy species, followed by a saturation phase, namely the To-phase, governed by convergence towards the TSS. This unified picture suggests that prion-like dynamics can be described, within each phase, by simpler analytical functions that retain a clear mechanistic interpretation. Such simple analytical functions are particularly attractive when the goal is not only to understand but also to predict disease progression.

A widely used hypothetical model of neurodegenerative disease progression is the sigmoidal trajectory of biomarkers proposed by Jack et al.~\cite{jack2010hypothetical,jack2013update}. Our results provide a mechanistic foundation for this hypothetical model within our class of heterodimer models based on the prion-like hypothesis, showing how exponential growth of the toxic species in the Mi-phase transitions naturally into logistic saturation in the To-phase. While statistical approaches have been successfully used to quantify correlations between biomarkers—such as the precedence of tau accumulation over amyloid $\beta$ demonstrated by Therneau et al.~\cite{therneau2021relationships}—they are inherently limited in their ability to capture underlying biological mechanisms and to extrapolate beyond observed data. In contrast, mechanistic models derived from scientific laws, such as the law of mass action, offer interpretable parameters and predictive capabilities grounded in biological processes.

This distinction is particularly relevant in light of recent advances in longitudinal biomarker measurements, such as repeated plasma p-tau217 sampling~\cite{kirsebom2025repeated}, which reveal heterogeneous yet predictable trajectories of disease progression. In such settings, mechanistic models provide a natural framework for forecasting, as they capture the underlying dynamical laws governing biomarker evolution. Combined with modern inference frameworks such as PINTS~\cite{clerx2019Pints}, this opens the possibility of calibrating prion-like heterodimer models to patient-specific data and predicting future disease trajectories. The HeMiTo framework therefore not only explains the qualitative phases of disease progression, but also provides a principled basis for quantitative prediction. Developing and validating such forecasting approaches in clinical settings represents a promising direction for future work.

\section*{Funding}
JGB is funded by a grant from the Wenner-Gren foundations (Grant number: FT2023-0005).



\begin{thebibliography}{10}

\bibitem{prusiner1998prion}
Stanley~B Prusiner, Michael~R Scott, Stephen~J DeArmond, and Fred~E Cohen.
\newblock Prion protein biology.
\newblock {\em cell}, 93(3):337--348, 1998.

\bibitem{alyenbaawi2020prion}
Hadeel Alyenbaawi, W~Ted Allison, and Sue-Ann Mok.
\newblock Prion-like propagation mechanisms in tauopathies and traumatic brain
  injury: challenges and prospects.
\newblock {\em Biomolecules}, 10(11):1487, 2020.

\bibitem{thompson2020protein}
Travis~B Thompson, Pavanjit Chaggar, Ellen Kuhl, Alain Goriely, and
  Alzheimer’s Disease~Neuroimaging Initiative.
\newblock Protein-protein interactions in neurodegenerative diseases: A
  conspiracy theory.
\newblock {\em PLoS computational biology}, 16(10):e1008267, 2020.

\bibitem{lemarre2020unifying}
Paul Lemarre, Laurent Pujo-Menjouet, and Suzanne~S Sindi.
\newblock A unifying model for the propagation of prion proteins in yeast
  brings insight into the [psi+] prion.
\newblock {\em PLoS computational biology}, 16(5):e1007647, 2020.

\bibitem{parsons:hal-04178969}
Todd~L Parsons and David J~D Earn.
\newblock {Analytical approximations for the phase plane trajectories of the
  SIR model with vital dynamics}.
\newblock working paper or preprint, August 2023.

\bibitem{borgqvist2025HeMiTo}
Johannes~G Borgqvist and Christoffer Gretarsson~Alexandersen.
\newblock Hemito-dynamics: a characterization of mammalian prion toxicity using
  non-dimensionalization, linear stability and perturbation analyses.
\newblock {\em Mathematical Medicine and Biology: A Journal of the IMA},
  42(2):159--175, 11 2024.

\bibitem{gerlee2022weak}
Philip Gerlee.
\newblock Weak selection and the separation of eco-evo time scales using
  perturbation analysis.
\newblock {\em Bulletin of Mathematical Biology}, 84(5):52, 2022.

\bibitem{fornari2019prion}
Sveva Fornari, Amelie Sch{\"a}fer, Mathias Jucker, Alain Goriely, and Ellen
  Kuhl.
\newblock Prion-like spreading of alzheimer’s disease within the brain’s
  connectome.
\newblock {\em Journal of the Royal Society Interface}, 16(159), 2019.

\bibitem{jack2010hypothetical}
Clifford~R Jack, David~S Knopman, William~J Jagust, Leslie~M Shaw, Paul~S
  Aisen, Michael~W Weiner, Ronald~C Petersen, and John~Q Trojanowski.
\newblock Hypothetical model of dynamic biomarkers of the alzheimer's
  pathological cascade.
\newblock {\em The Lancet Neurology}, 9(1):119--128, 2010.

\bibitem{jack2013update}
Clifford~R Jack~Jr, David~S Knopman, William~J Jagust, Ronald~C Petersen,
  Michael~W Weiner, Paul~S Aisen, Leslie~M Shaw, Prashanthi Vemuri, Heather~J
  Wiste, Stephen~D Weigand, et~al.
\newblock Update on hypothetical model of alzheimer’s disease biomarkers.
\newblock {\em Lancet neurology}, 12(2):207, 2013.

\bibitem{therneau2021relationships}
Terry~M Therneau, David~S Knopman, Val~J Lowe, Hugo Botha, Jonathan
  Graff-Radford, David~T Jones, Prashanthi Vemuri, Michelle~M Mielke,
  Christopher~G Schwarz, Matthew~L Senjem, et~al.
\newblock Relationships between $\beta$-amyloid and tau in an elderly
  population: An accelerated failure time model.
\newblock {\em Neuroimage}, 242:118440, 2021.

\bibitem{kirsebom2025repeated}
Bj{\o}rn-Eivind Kirsebom, Fernando Gonzalez-Ortiz, Sinthujah Vigneswaran, Geir
  Br{\aa}then, Ragnhild~Eide Skogseth, Berglind G{\'\i}slad{\'o}ttir, Peter
  Harrison, Jonas~Alexander Jarholm, Lene P{\aa}lhaugen, Arvid Rongve, et~al.
\newblock Repeated plasma p-tau217 measurements to monitor clinical progression
  heterogeneity.
\newblock {\em Alzheimer's \& Dementia}, 21(5):e70319, 2025.

\bibitem{clerx2019Pints}
Michael Clerx, Martin Robinson, Ben Lambert, Chon~Lok Lei, Sanmitra Ghosh,
  Gary~R Mirams, and David~J Gavaghan.
\newblock Probabilistic inference on noisy time series ({PINTS}).
\newblock {\em Journal of Open Research Software}, 7(1):23, 2019.

\end{thebibliography}
\end{document}